\documentclass[12pt]{amsart}

\usepackage{amssymb,amsmath,amsthm}
\usepackage{graphicx}
\usepackage{comment}
\usepackage[numbers]{natbib}

\newtheorem{thm}{Theorem}[section]
\newtheorem{lem}[thm]{Lemma}

\theoremstyle{observation}
\newtheorem{remark}{Remark}
\theoremstyle{remark}

\begin{document}
\title[Bounds on odd cycles in the Incompatibility Graph]{A tight bound on the length of odd cycles in the incompatibility
  graph of a non-C1P matrix}
\address{Department of Mathematics, Simon Fraser University, Burnaby, BC,
   Canada}
\author[Mehrnoush Malekesmaeili]{Mehrnoush Malekesmaeili}
\email{\{mmalekes,cedric.chauve,tamon\}@sfu.ca}
\author[Cedric Chauve]{Cedric Chauve}
\author[Tamon Stephen]{Tamon Stephen}

\maketitle

\begin{abstract}
  A binary matrix has the consecutive ones property (C1P) if it is
  possible to order the columns so that all 1s are consecutive in
  every row.  In [McConnell, SODA 2004 768-777] the notion of
  incompatibility graph of a binary matrix was introduced and it was
  shown that odd cycles of this graph provide a certificate that a
  matrix does not have the consecutive ones property. A bound of $k+2$
 was claimed for the smallest odd cycle of a non-C1P matrix
  with $k$ columns. In this note we show that this result can
  be obtained simply and directly via Tucker patterns, and that
  the correct bound is $k+2$ when $k$ is even, but $k+3$ when $k$
  is odd.
\end{abstract}



\section{Introduction}
\label{sec1}
A binary matrix has the \textit{Consecutive Ones Property} (C1P), if there
exists a permutation of its columns
that makes the 1s consecutive in every row. 
  It was first introduced by Fulkerson
and Gross in \citep{Fulkerson1965} as special case of deciding whether a
graph is an interval graph, and has important
applications in computational biology, see for example
 \citep{Alizadeh1994}.
The problem of deciding whether a given
binary matrix has the C1P can be solved efficiently
\citep{Booth1976}.
Clearly, checking the claim that a matrix is C1P is easy, provided a
valid permutation of the columns is given. However it is not 
obvious how to certify that a matrix is not C1P. 

The first structural result on non-C1P matrices is due to Tucker, who
proved in \citep{Tucker1972} that a binary matrix is not C1P if and
only if it contains a submatrix from one of five families of binary
matrices known as \textit{Tucker patterns}, that define then a natural
family of certificates for non-C1P matrices.  Tucker patterns can be
detected in polynomial time \citep{Dom2010}.

McConnell, in \citep{McConnell2004}, defined a simple and elegant
certificate for non-C1P matrices.  He introduced the notion of the
\textit{incompatibility graph} of a binary matrix, and proved that a
matrix is C1P if and only if this graph is bipartite. Hence, an odd
cycle in this graph is a non-C1P certificate. He claimed that the
incompatibility graph of a non-C1P matrix with $k$ columns always has
an odd cycle of length at most $k+2$ and proposed a linear time
algorithm to compute such a cycle, more efficient than the currently
best algorithms to detect Tucker patterns \citep{Dom2010}.

In this note we correct the bound McConnell gave in
\citep{McConnell2004} for the length of the odd cycle of the
incompatibility graph. We prove that the incompatibility
graph of a non-C1P matrix with $k \ge 4$ columns 
has the smallest odd cycle of length at most $k+2$ if $k$ is
odd and $k+3$ if $k$ is even, and that this 
bound is tight. Our approach
relies on investigating the odd cycles of the incompatibility graphs
of Tucker patterns.

\section{Preliminaries}
\subsection{Tucker patterns}
Tucker characterized C1P matrices via minors known as ``Tucker \newline patterns'' 
illustrated in Figure~\ref{fig:tucker}.
\begin{figure}
\begin{center}
\includegraphics[scale=0.27]{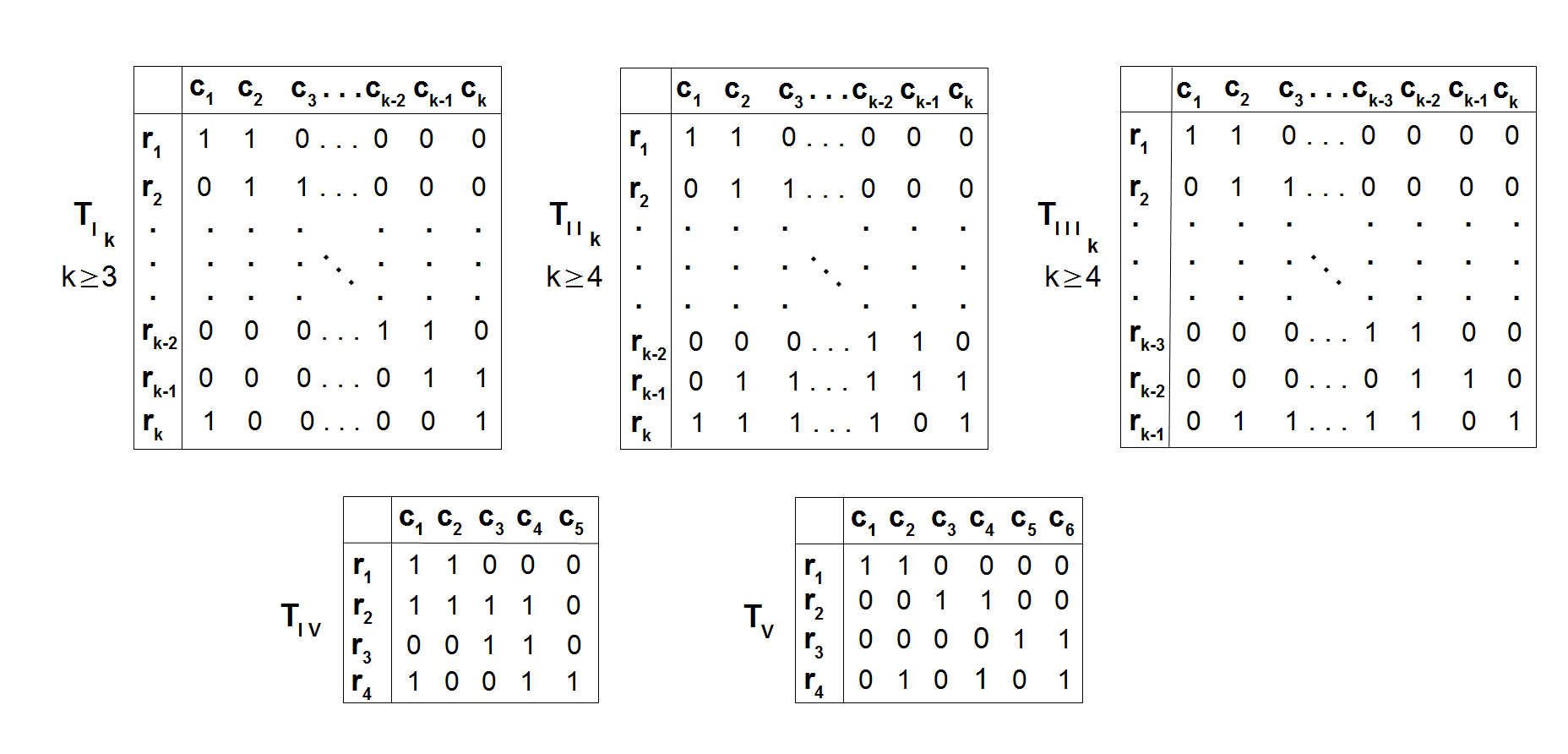}
\caption{The five Tucker patterns}
\label{fig:tucker}
\end{center}
\end{figure}

\begin{lem}\citep{Tucker1972}
A binary matrix has the Consecutive Ones Property if and only if it contains
none of the five Tucker patterns as a submatrix.
\end{lem}
These patterns are the minimal structures that
obstruct the matrix $M$ from having the C1P, i.e.~removing a row from these
structures result in a C1P-matrix. In this present note, for all of the 
five Tucker matrices we consider the order of columns given
in Figure~\ref{fig:tucker}.
\subsection{Incompatibility graph}
The concepts of incompatibility and forcing graphs were introduced by McConnell
in \citep{McConnell2004}. 
Let $M$ be an $m \times n$ binary matrix with rows
$R=\{r_1,r_2,\ldots,r_{m}\}$ and columns $C=\{c_1,c_2,\ldots$
$,c_n\}$. The \textit{incompatibility graph} of  $M$ is an undirected graph
$G_M=(V,E)$,
whose vertices are pairs $(c_i,c_j)$ 
(for $i,j=1,\ldots,n, i \neq j$).  
Two vertices $(c_i,c_j)$ and $(c_j,c_k)$ are
adjacent, if one of the following holds:
\begin{enumerate}
 \item $c_i=c_k$.
\item There exists a row $r_l$ in $M$ such that $M_{li},M_{lk}=1$ but
$M_{lj}=0$.
\end{enumerate}
We sometimes refer to these edges as type 1 or type 2,
as appropriate.
These edges represent incompatible pairs of orderings, i.e.~each
edge corresponds to two relative orderings of the columns that
cannot appear simultaneously in a consecutive ones ordering of
the matrix.  McConnell noted that the incompatibility graph is
bipartite if and only if the matrix is C1P.  Thus odd cycles
in the incompatibility graph certify that a matrix is not C1P.

The \textit{forcing graph} $F_M=(V,E^{'})$ is an undirected graph 
whose vertex set is
same as that of $G_M$ and whose edge set is a set of all pairs 
$((c_i,c_j),(c_k,c_j))$ where $((c_i,c_j),(c_j,c_k))$ is an edge of $G_M$.
It is not hard to see that the incompatibility graph and the
forcing graph both have $n(n-1)$ vertices and are symmetric. 
As with the incompatibility graph, the forcing graph can be used
to certify that a matrix is not C1P: a path in this graph from
$(c_i,c_j)$ to $(c_j,c_i)$ for and $i,j,$ represents a chain of 
implications (``forcings'') leading to a contradiction.

In fact, McConnell observed that these certificates are almost
the same: such a path in
the forcing graph can be transformed to a cycle in the
incompatibility graph and vice versa. 
\begin{lem}\label{lem2.2}
If there exists a path with $m$ vertices in $F_M$ between 
$(c_i,c_j)$ and $(c_j,c_i)$ there is an odd cycle of length $m-1$ 
(if $m$ is even) or $m$ (if $m$ is odd) in $G_M$ 
containing the vertex $(c_i,c_j)$.  Conversely, if there is 
an odd cycle of length $m$ in $G_M$ containing $(c_i,c_j)$
there is a path with at most $m+1$ vertices in $F_M$ from 
$(c_i, c_j)$ to $(c_j, c_i)$.
\end{lem}
\begin{proof}

Without loss of generality the path in $F_M$ is:
$$P: \quad (v_1=(c_i,c_j), v_2,v_3, v_4,\ldots,v_{m}=(c_j, c_i)).$$
Let $v_k^{'}=(c',c)$ when $v_k=(c,c')$.  Then in $G_M$ we can build
the walk:
$$P': \quad (v_1,v_2',v_3,v_4',\ldots,v_{m}^{(')})$$
When $m$ is even, the final vertex in this walk is $v_m'=v_1$, and we have
a cycle with $m-1$ vertices.  When $m$ is odd, the final term is $v_m$,
and we can complete the cycle using the type 1 edge $(v_m,v_m'=v_1$).  
In this case the odd cycle has length $m$.  

Similarly, an odd cycle of length $m$ in $G_M$ can be transformed 
into a walk of length $m+1$ in $F_M$ by performing the reverse
operation on the walk with even length $m+1$
 by taking the
cycle vertices starting an ending in $v_1$.  
Note that if this path contains some
type 1 edge $(v_k,v_k')$ in $G_M$, this becomes a trivial edge
$(v_k, v_k)$ or $(v_k', v_k')$ in $F_M$ and should be contracted,
reducing the length of the found path.
\end{proof}

Given $G_M$ ($F_M$), we define $G_M^1$ ($F_M^1$) and $G_M^2$ ($F_M^2$) 
to be the subgraphs induced by the vertex sets 
$V_1=\{(c_i,c_j)~|~ i<j\}$ and $V_2=\{(c_i,c_j)~|~ i>j\}$ respectively. 
We observe that the two pairs of subgraphs are isomorphic. 

Suppose now that we build $G_M$ and $F_M$ graph for a given
Tucker pattern $M$ from Figure~\ref{fig:tucker} by first generating
the type 1 edges of $G_M$ and then adding the edges generated by
each row in turn, beginning at the top.  
Edges $e=((c_i, c_j), (c_l, c_j))$ and $e^{'}=((c_j, c_i), (c_j, c_l))$
in $F_M$ are generated by triples $(i,j,l)$ from
a given row $r$ exactly when $M_{ri}, M_{rl}=1$  but $M_{rj}=0$.
The edges corresponding to the rows from the top of the matrix
then come in pairs $e,~e^{'}$, where one is contained in $F_M^1$ 
and the other in $F_M^2$. 
As we descend the rows, the ones in the rows are consecutive
until we reach the final row, $r_t$, that has gaps between its 1s entries. 

The edges of $F_M$ generated by the gaps in $r_t$, i.e.~triples $(i,j,l)$
of columns where $i<j<l$ are the only edges which go between
$V_1$ and $V_2$.  We call these edges {\it critical}.

   \section{Finding odd cycles using Tucker configurations}\label{se:proofs}
We now give a tight bound on the smallest odd cycle in the incompatibility graph
of a non-C1P matrix using Tucker matrices.

\begin{thm}\label{thm3.1}
  The length of the smallest odd cycle in the incompatibility graph of a binary
matrix with $k \ge 4 $ columns is at most $k+2$ if $k$ is odd
or $k+3$ is $k$ is even, and this bound is tight.
\end{thm}
We begin by remarking that since any non-C1P matrix $M$ contains a
Tucker pattern as a submatrix, we can restrict our attention to
Tucker patterns when looking for short odd-cycles in the 
incompatibility graph.  This is because if we look at the
subgraph of the incompatibility graph induced by considering only
the columns (vertices) and rows (edges) of $M$ ($G_M$) containing the Tucker
pattern, we get exactly the incompatibility graph of the submatrix,
which has at most as many columns as $M$.  So the upper bound
for Tucker patterns holds for all $M$, and the worst case for a
given number of columns will occur at a Tucker pattern.  

We remark that for $k \le 2$ all binary matrices have the C1P, and
for $k=3$ if a matrix is not C1P it must contain the Tucker pattern
$T_{I_3}$ as a submatrix, and thus have an odd cycle of length 3
in its incompatibility graph, see Section~\ref{se:pat1}.
For $k \ge 4$, the tight bound of $k+2$ or $k+3$ is attained by $T_{III_k}$,
see Section~\ref{se:pat3}.  We proceed to analyze each Tucker
pattern separately.

\subsection{First Tucker pattern}\label{se:pat1}
$T_{I_k}$ is shown in Figure~\ref{fig:tucker}; it is a square matrix 
of size $k$ where $k \geq 3$.
\begin{lem}\label{lem3.2}
 For $k \ge 3$, 
 the length of the smallest odd cycle in the incompatibility graph of 
 $T_{I_k}$ is $k$ when $k$ is odd and $k+1$ when $k$ is even.
\end{lem}
\begin{proof}
We find a path in $F_{M_{1_k}}$ in $T_{I_k}$ between $(c_1, c_{k-1})$
and $(c_{k-1}, c_1)$.
Since $M_{11},M_{12}=1$ but $M_{1~k-1}=0$, we have that
$((c_1,c_{k-1}),(c_2,c_{k-1}))$ is an edge of $F_{M_{I_k}}$.
Similarly $((c_i,c_{k-1}),(c_{i+1},c_{k-1}))$ is an edge of $F_{M_{I_k}}$ for
$i=2,\ldots,k-3$ using row $i$ of $M$.
Using row $k-2$, we get that $(c_{k-2},c_{k-1})$ forces $(c_{k-2},c_{k})$ and
using row $k-1$ that $(c_{k-2},c_{k})$ forces $(c_{k-1},c_{k})$. 
Observe that $e=((c_1,c_{k-1}),(c_{k}, c_{k-1}))$ is a critical edge
of $F_{M_{1_k}}$. 
Therefore $(c_1,c_{k-1}),(c_2,c_{k-1}),(c_3,c_{k-1}),\ldots,
(c_{k-2},c_{k-1}),(c_{k-2},c_{k}), (c_{k-1},c_{k}),(c_{k-1},c_1)$ 
is a path with $k+1$ vertices in $F_{M_{I_k}}$. 
By Lemma~\ref{lem2.2}, this gives the required cycle.

Finally, we note that if there is any odd cycle in the
incompatibility graph of length less than $k$,
we would derive a contradiction to the C1P using fewer than $k$
columns, contradicting the minimality of the Tucker pattern.
Thus the length of this odd cycle is in fact minimal.
\end{proof}
\subsection{Second Tucker pattern}\label{se:pat2}
For $k \ge 4$, $T_{II_{k}}$ is a square matrix of size $k$.
We use a strategy to find a cycle in the incompatibility
graph of $T_{II_{k}}$ that is similar to that of $T_{I_{k}}$.
\begin{lem}\label{lem3.3}
 The smallest odd cycle in the incompatibility graph of $T_{II_{k}}$ has length
$k$ when $k$ is odd and $k+1$ when $k$ is even.
\end{lem}
\begin{proof}
From row $i$ of the matrix for $i=1,\ldots,k-2$, we get that $(c_i,c_{k})$ 
forces $(c_{i+1},c_k)$.  From row $k-1$, $(c_1,c_{k})$ forces $(c_1,c_{k-1})$.
Finally, $((c_{1},c_{k-1}),(c_{k}, c_{k-1}))$ is a critical edge.  Then
$(c_1,c_{k}),(c_2,c_k),\ldots, (c_{k-1},c_{k}),(c_{k-1},c_1), (c_k, c_1)$ is
a path of length $k+1$ in the 
forcing graph of $T_{II_{k}}$. Using Lemma~\ref{lem2.2} we can find an odd cycle
of length
either $k$ or $k+1$ 
containing all rows of the pattern. 
Again the minimality of the Tucker pattern ensures that we cannot
have a cycle of length less than $k$.
\end{proof}
\subsection{Third Tucker pattern}\label{se:pat3}
Now we consider the third Tucker pattern that has $(k-1)$ rows 
and $k$ columns where $k \geq 4$.
\begin{lem}\label{lem3.4}
The smallest odd cycle in the incompatibility graph of the third Tucker pattern
has length $k+2$ if $k$ is odd and $k+3$ if $k$ is even. 
\end{lem}
\begin{proof}
In this case, because we need to prove a non-trivial lower
bound, we will describe the full structure of $F_{M_{III}}$.
The graph $F_{M_{III_6}}$ is illustrated in
Figure~\ref{fig:tiii} and captures the features we are interested in.
\begin{figure}
\begin{center}
\includegraphics[scale=0.4]{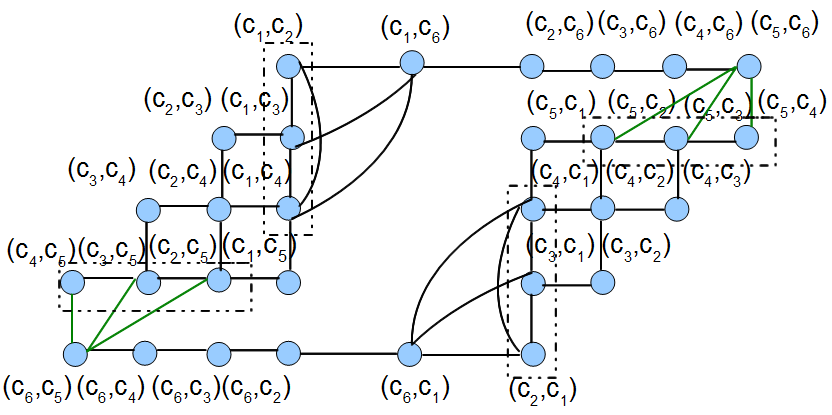}
\caption{$F_{M_{III_6}}$}
\label{fig:tiii}
\end{center}
\end{figure}
Consider first the portion of the graph generated by 
excluding the last row and column on $M_{III_k}$.
In this case each row has a unique pair of 1 entries;
these can be combined with any zero entry to get a forcing
triple.  The result is a triangular grid on $V_1$
(and symmetrically, $V_2$), where vertex
$(c_i,c_j)$ is connected to all of $(c_{i \pm 1}, c_j)$
and $(c_i, c_{j \pm 1})$ that are also vertices of $V_1$ 
with coordinates between 1 and $k-1$ and in increasing order.
Now, returning our attention the last column, we see it generates a path
$(c_1, c_k), (c_2, c_k), \ldots (c_{k-1}, c_k)$ by considering the
pair of ones in each row in turn..
These two components, and their symmetric copies in $V_2$
are the entirety of the graph if we exclude the final row.

The first zero from the final row combines with the many
pairs of ones to connect the two components of $V_1$
by fusing $(c_1, c_2), (c_1, c_3), \ldots (c_1, c_{k-2})$
and $(c_1, c_k)$ (but not $(c_1, c_{k-1})$) into a clique.
Finally the second zero (in row $k-1$) produces the
critical edges $(c_i, c_k), (c_k, c_{k-1})$ for 
$i=2,\ldots,k-2$ and fuses all these vertices into a
clique.  

We can then see that a path with $(k+3)$ vertices in
$F_{M_{III_k}}$ from $(c_1, c_k)$ to $(c_k, c_1)$ is given by:
$(c_1, c_k), (c_2, c_k), \ldots, (c_{k-1}, c_k),
 (c_{k-1}, c_2), (c_{k-1}, c_1), (c_{k-2}, c_1), (c_k, c_1).$
By Lemma~\ref{lem2.2}, this gives the required cycle in
the incompatibility graph of length $(k+3)$ if $k$ is even,
and $(k+2)$ if $k$ is odd.

To prove that this is shortest, it suffices to show that this
is the shortest path between $(c_i, c_j)$ and $(c_j, c_i)$ 
for some $i,j$ in $F_{M_{III_k}}$, since if there is
a shorter odd cycle in the incompatibility graph
of length $(k+1)$ with $k$ even, or $k$ with $k$ odd,
by Lemma~\ref{lem2.2} there would be a path of length
at most $k+2$ between some $(c_i, c_j)$ and $(c_j, c_i)$.

We can see that there is no shorter path by contracting
the groups of vertices illustrated in Figure~\ref{fig:tiii},
i.e.~$(c_1, c_j)$ for $j=2, 3, \ldots k-2$; 
$(c_i, c_{k-1})$ for $i=2, 3, \ldots k-2$; and the
symmetric groups on $F_{V_2}$.  This will not increase the
distance between any pair of vertices.
We can see that what remains is $2k-2$ cycle with vertices
$(c_i, c_j)$ opposite $(c_j, c_i)$, an additional part
of $F_{V_1}$ attached to the two contracted vertices,
and symmetrically in $F_{V_2}$.  The shortest path
between opposite vertices on the cycle has $(k+3)$
vertices, as does the shortest path between any of
additional vertices in $F_{V_1}$ and those of $F_{V_2}$,
though for some choices of $(i,j)$ the shortest path from
$(c_i, c_j)$ to $(c_j, c_i)$ may be longer.
\end{proof}

Taking $k \ge 4$ even,
this gives a family of counter examples to Theorem 6.1 of McConnell in
\citep{McConnell2004}. For example, taking $k=4$, we have
$M=\bigl(\begin{smallmatrix}
1&1&0&0\\ 0&1&1&0\\0&1&0&1\\
\end{smallmatrix} \bigr)$.
Then $F_M$ is a 12-cycle and $G_M$ is a 12-cycle with 
6 chords added between opposite vertices of the cycle. It is clear that the
smallest odd cycle is of length 7. This graph is shown in
Figure~\ref{fig:tiii4}.
\begin{figure}
\begin{center}
\includegraphics[scale=0.2]{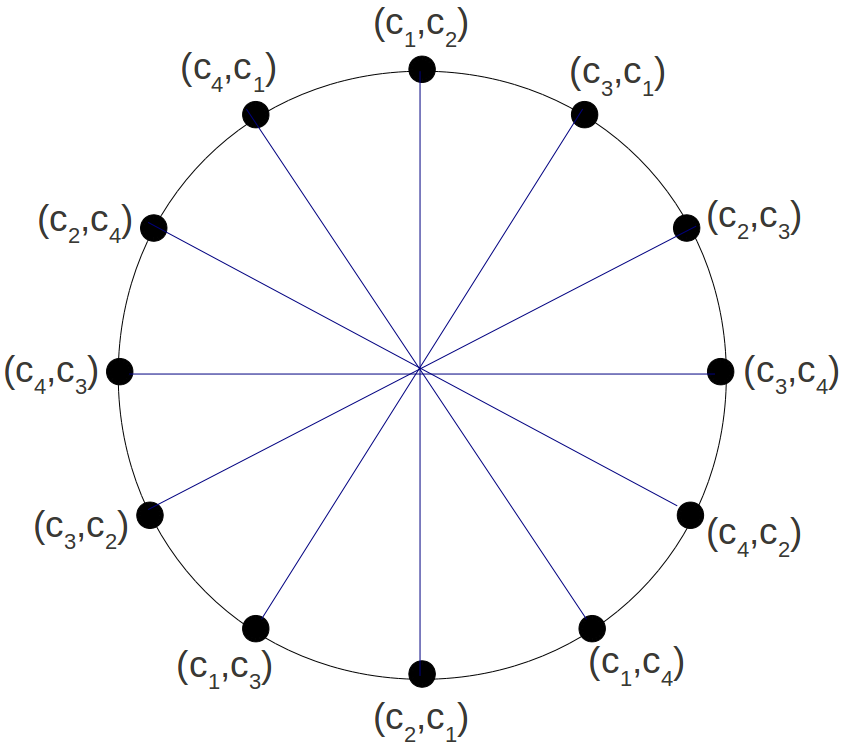}
\caption{Incompatibility Graph of $T_{III_{4}}$}
\label{fig:tiii4}
\end{center}
\end{figure}

\subsection{Fourth Tucker pattern}\label{se:pat5}
The fourth Tucker pattern is of size 4 by 5,
with $((c_1,c_3),$ $(c_5,c_3))$ as a critical edge of $F_{M_{IV}}$.  Here
$(c_1,c_3)$ forces $(c_5,c_3)$, $(c_5,c_3)$ forces $(c_5,c_2)$, $(c_5,c_2)$
forces $(c_4,c_2)$, $(c_4,c_2)$ forces $(c_3,c_2)$ and $(c_3,c_2)$ forces $(c_3,c_1)$, which gives a
path with $6
$ vertices in $F_{M_{IV}}$
and an odd cycle of length $5$ in the incompatibility graph of $T_{IV}$. 
\subsection{Fifth Tucker pattern}\label{se:pat4}
The fifth Tucker pattern shown in Figure \ref{fig:tucker} is 4 by 6.
It can be observed that $((c_2,c_3),(c_6,c_3))$ is a critical edge
of $F_{M_V}$.
Now
$(c_2,c_3)$ forces $(c_6,c_3)$ which forces $(c_6,c_4)$; $(c_6,c_4)$ forces
$(c_5,c_4)$, $(c_5,c_4)$ forces
$(c_5,c_2)$, $(c_5,c_2)$ forces $(c_5,c_1)$ which forces $(c_6,c_1)$. $(c_6,c_1)$ forces $(c_4,c_1)$
which forces $(c_4,c_2)$; $(c_4,c_2)$ forces $(c_3,c_2)$. This gives a path with $10$ vertices in $F_{M_V}$ and an odd cycle of length $9$ in $T_V$.
In this case the length of the smallest odd cycle also attains the
bound of Theorem~\ref{thm3.1}.

Combining these five case allows us to conclude Theorem~\ref{thm3.1}. 
\begin{remark}
Running the partition refinement algorithm of \citep{McConnell2004}
on $M_{III_k}$ may generate an odd cycle of length as much as
$2k-1$ for the certificate, depending on which critical edge
is processed from the last row.  
\end{remark}


\section{Acknowledgments}
All three authors were supported by NSERC Discovery Grants.
The authors thank Ross McConnell for helpful discussion.



\bibliographystyle{amsalpha}
\bibliography{Ref}

\end{document}